\documentclass[12pt,a4paper]{article}
\usepackage{amsmath,amsthm,amsfonts,amssymb,bbm}
\usepackage{graphicx,psfrag,subfigure,color}
\usepackage{cite}
\usepackage{hyperref}

\numberwithin{equation}{section}

\newcommand{\Or}{\mathcal{O}}

\newcommand{\Pb}{\mathbb{P}}

\newcommand{\Id}{\mathbbm{1}}
\newcommand{\e}{\varepsilon}
\newcommand{\I}{{\rm i}}

\newcommand{\R}{\mathbb{R}}

\newcommand{\Z}{\mathbb{Z}}
\renewcommand{\Re}{\mathrm{Re}}
\renewcommand{\Im}{\mathrm{Im}}

\DeclareMathOperator{\Tr}{Tr}

\newtheorem{prop}{Proposition}[section]
\newtheorem{thm}[prop]{Theorem}
\newtheorem{lem}[prop]{Lemma}
\newtheorem{defin}[prop]{Definition}
\newtheorem{cor}[prop]{Corollary}

\newtheorem{cla}[prop]{Claim}

\newtheorem{rem}[prop]{Remark}

\title{Finite GUE distribution with cut-off at a shock}
\author{P.L. Ferrari\thanks{Institute for Applied Mathematics, Bonn University, Endenicher Allee 60, 53115 Bonn, Germany. E-mail: {\tt ferrari@uni-bonn.de}}}
\date{March 3, 2018}

\begin{document}
\sloppy
\maketitle

\begin{abstract}
We consider the totally asymmetric simple exclusion process with initial conditions generating a shock. The fluctuations of particle positions are asymptotically governed by the randomness around the two characteristic lines joining at the shock. Unlike in previous papers, we describe the correlation in space-time \emph{without} employing the mapping to the last passage percolation, which fails to exists already for the partially asymmetric model. We then consider a special case, where the asymptotic distribution is a cut-off of the distribution of the largest eigenvalue of a finite GUE matrix. Finally we discuss the strength of the probabilistic and physically motivated approach and compare it with the mathematical difficulties of a direct computation.
\end{abstract}

\section{Introduction}\label{sectIntro}
The totally asymmetric simple exclusion process (TASEP) is one of the simplest non-reversible interacting particle system. The occupation variables of the TASEP are denoted by $\eta_j$, $j\in\Z$, where $\eta_j=1$ if site $j$ is occupied by a particle and $\eta_j=0$ if it is empty. The time evolution of  TASEP is the following. Particles jump one step to the right and are allowed to do so only if their right neighboring site is empty. Jumps are independent of each other and are performed after an exponential waiting time with mean $1$, which starts from the time instant when the right neighbor site is empty.

Since particles can not overtake each other, we can associate a labeling to them and denote the position of particle $k$ at time $t$ by $x_k(t)$. We choose the right-to-left ordering, namely, we denote by $x_{k+1}(0)< x_k(0)$ for any $k$. For any initial condition with $N$ particles at non-random positions, $x_N(0)<x_{N-1}(0)<\ldots<x_1(0)$, by Theorem~2.1 of~\cite{BFPS06}, we know that
\begin{equation}\label{eq1}
\Pb(x_N(t)\geq x)=\det(\Id-K_{N,t})_{\ell^2((-\infty,x))}
\end{equation}
for some correlation kernel $K_{N,t}$ depending on $\{x_n(0),n=1,\ldots,N\}$ and $t$ (and similarly for joint distributions of particles).

To determine the correlation kernel one has to do a biorthogonalization. This has been successfully carried out for several initial condition. Large time asymptotic analysis for different initial conditions, leads to the discovery of limiting processes as the Airy$_1$ or the Airy$_{2\to 1}$ transition process for TASEP in continuous or discrete time and some generalizations like PushASEP and/or with particle-dependent jump rates~\cite{Sas07,BFPS06,BFP06,BF07,BFS07b}. A Fredholm determinant expression for the distribution of particle positions is available also for some random initial conditions. For instance, using Burke's theorem~\cite{Bur56}, one can replace a Bernoulli-product measure with density $\alpha$ to the right of the origin with a single particle with reduced jump rate $1-\alpha$, see e.g.~\cite{PS01}. Finally, the case of Bernoulli-product measure to the left of the origin does also fit in the determinantal framework, see~\cite{FSW15}. Very recently, a smart way to do the biorthogonalization for general non-random initial, which can be used for asymptotic analysis, has been discovered~\cite{MQR17}.

The formula (\ref{eq1}) has been successfully analyzed in the large time limit whenever, under appropriate scaling limit, the kernel itself converges to a limiting kernel. However, in all the cases where the system generates a shock, i.e., a discontinuity in the particle density, a direct computation using (\ref{eq1}) does not lead to results due to the fact that $K_t$ did not converge to a limiting kernel (although the Fredholm determinant still being well-defined). A way out in these case was found in~\cite{FN13} where we reduced the analysis to simpler cases. The core of the proof is then probabilistic and based on the following two ingredients, reflecting physical behavior of the system:
\begin{itemize}
\item[(1)] A shock is a position where two characteristic lines meet. Further, positions of particles in space-time with time difference of order $t$ are non-trivially correlated if they are in a $t^{2/3}$ neighborhood of a given characteristics. This means that particles close to the shock at time $t$ are non-trivially correlated with two initial configurations, one of each side of the shock, as discussed in Section~\ref{SectAsymptIndep}.
\item[(2)] Along the characteristics, decorrelation occurs only over macroscopic time differences (called slow-decorrelation phenomenon~\cite{Fer07,CFP10b}). This means that the fluctuations of particles close to a given characteristic at time $t-o(t)$, which typically live on a $t^{1/3}$ scale, and the fluctuations of the particles close to the same characteristic at time $t$ will differ only by $o(t^{1/3})$. This is discussed in Section~\ref{SectSlowDec}.
\end{itemize}
In particular, if we take time $t-t^\nu$ with $\nu\in (2/3,1)$, then the particles close to the two characteristics which meet at time $t$ are at distance at least $\Or(t^\nu)\gg t^{2/3}$ and thus asymptotically independent, see Figure~\ref{FigShockTASEP} for an illustration.
\begin{figure}
\begin{center}
\psfrag{x}[cc]{$x$}
\psfrag{t}[lb]{$t$}
\psfrag{tau}[lb]{$t-t^\nu$}
\psfrag{E}[cc]{$E$}
\psfrag{El}[cb]{$E_\ell$}
\psfrag{Er}[cb]{$E_r$}
\includegraphics[height=5cm]{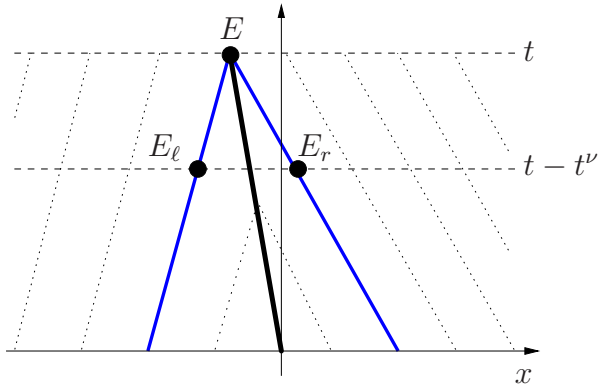}
\caption{Illustration of the characteristics for TASEP. $E=(x,t)$ is the shock location, where two characteristics with speed $v_\ell>v_r$ merge (the thick blue lines). Due to the slow decorrelation along characteristics, at large time $t$ the fluctuations at $E$ are inherited by the ones at $E_\ell=(x-v_\ell t^\nu,t-t^\nu)$ and $E_r=(x-v_r t^\nu,t-t^\nu)$ up to $o(t^{1/3})$ if $\nu<1$.}
\label{FigShockTASEP}
\end{center}
\end{figure}

This procedure of reduction to two simpler models, for which the distribution function was given by a Fredholm determinant with a convergent kernel, was so-far employed by using the useful connection with the last passage percolation (LPP) model. This leads to the results on the fluctuation of the particle positions around the shock~\cite{FN13}, the non-trivial results on the competition interface~\cite{FN16}, and the second class particles~\cite{FGN17} with \emph{non-random} initial conditions. Whenever the fluctuations of the bulk of the LPP are on a smaller scale than the boundary terms, or one on the two LPP dominates the other, the analysis is simpler and one can apply the bootstrap argument of~\cite{BC09}.

The procedure illustrated above has been used by Nejjar in~\cite{N17} to show decoupling under the double scaling limit, where one first takes the large time limit with a microscopic shock (of size $\beta t^{-1/3}$ so that the kernel is well-defined~\cite{FN14}) and then $\beta\to\infty$. Very recently, Quastel and Rahman managed through a nice decomposition to handle the problem of the double scaling limit from a Fredholm determinant approach~\cite{QR18}.

The map from TASEP to LPP is nice since one has a geometric picture instead of a dynamic one. However, since this map holds only for totally asymmetry of the jumps, one would like to be able to describe the space-time correlations without the LPP picture. On the other hand, slow-decorrelation still holds for the partially asymmetric simple exclusion process (PASEP)~\cite{CFP10b}. Its proof is essentially based on the attractiveness of the model. In Section~\ref{SectSplitting} we show how the splitting into two simpler models follows from the basic coupling \emph{without using the LPP mapping}, and thus leading the way to generalizations to models like PASEP where the LPP mapping is not available.

To illustrate how to obtain results on the correlation in space-time, we consider a concrete case for which we can prove a new result as well. We consider particles with jump rate $1$ starting from every even site of $\Z_-$ and $M$ particles with jump rate $\alpha$ densely packed to the right of the origin. For $\alpha<1/2$ a shock is created. For $M=1$ the limiting distribution of particles around the shock were stated in Proposition~1 of~\cite{BFS09}. The proof was however not complete and if one tries to work out the details one sees that the kernel is not converging to a limit. We explain the mathematical difficulties of a direct computation for this case in Section~\ref{SectComparison}. Theorem~\ref{ThmMain} is the generalization to $M$ slow particle of Proposition~1 of~\cite{BFS09}. Its simple proof is in Section~\ref{SectProofMainThm}.

\subsubsection*{Acknowledgments}
P.L.~Ferrari is supported by the German Research Foundation in the Collaborative Research Center 1060 "The Mathematics of Emergent Effects", project B04.

\section{Model and main result}
We consider TASEP with initial condition
\begin{equation}\label{eqICx}
x_n(0)=
\left\{\begin{array}{ll}
-2n, & n\geq 1,\\
-n,& 0\geq n\geq -M+1,
\end{array}
\right.
\end{equation}
and where particle with label $n$ has jump rate $v_n$ with
\begin{equation}\label{eqICv}
v_n=\left\{\begin{array}{ll}
1, & n\geq 1,\\
\alpha,& 0\geq n\geq -M+1,
\end{array}
\right.
\end{equation}
with $\alpha\in (0,1/2)$.

In this case the slow particles create a shock with density $1-\alpha$ inside the jammed region and to the left of the shock the density of particle is $1/2$ as at time $0$. A simple macroscopic computation gives that the speed of the shock is $\mathbf{v}=\alpha-1/2$. Already in the $M=1$ case, which is equivalent to stationary initial condition with density $1-\alpha$ to the right of the origin by Burke's theorem, we expect to have fluctuations of particles inside the shock of order $t^{1/2}$, while to its left only $t^{1/3}$. Finally, the distribution of the left-most slow particle converges under diffusion scaling limit to the one of the largest eigenvalue of a $M\times M$ GUE matrix, see e.g.~\cite{FF12b,GS12}.

\begin{defin}
A GUE(M) random matrix is a $M\times M$ matrix $H$ distributed according to
\begin{equation}
\Pb(H\in dH)={\rm const} e^{-\frac12 \Tr(H^2)} dH,
\end{equation}
with $dH=\prod_{i=1}^M dH_{i,i}\prod_{1\leq j<k\leq M} d\Re(H_{j,k}) d\Im(H_{j,k})$.
\end{defin}
It is well-known that the eigenvalues of a GUE(M) matrix form a determinantal point process, see e.g.~Sect.~5.2 of~\cite{Meh91}. In particular, the distribution of its largest eigenvalue is a Fredholm determinant.
\begin{prop}
The distribution of the largest eigenvalue of a GUE(M) random matrix is given by
\begin{equation}
F_{\rm GUE(M)}(s)=\det(\Id-K_{\rm GUE(M)})_{L^2((s,\infty),dx)}
\end{equation}
where the correlation kernel is given by
\begin{equation}
K_{\rm GUE(M)}(s_1,s_2)=\frac{1}{(2\pi\I)^2}\oint_{|z|=\e/2} dz \int_{\e+\I\R} dw \frac{e^{w^2/2-w s_1}}{e^{z^2/2-z s_2}}\frac{w^M}{z^M}\frac{1}{w-z},
\end{equation}
for any $\e>0$.
\end{prop}

The generalization of Proposition~1 of~\cite{BFS09} is the following.
\begin{thm}\label{ThmMain}
Consider TASEP with initial conditions (\ref{eqICx}) and jump rates (\ref{eqICv}) with $\alpha<1/2$. Define the constants $\sigma=\sqrt{\frac{\alpha(1-2\alpha)}{2(1-\alpha)}}$, $\xi_c=\eta\frac{1-2\alpha}{1-\alpha}/\sigma=\eta \sqrt{\frac{2(1-2\alpha)}{\alpha(1-\alpha)}}$ and consider the scaling
\begin{equation}\label{eq2.6}
n=\frac{1-\alpha}{2}t+\eta t^{1/2},\quad x(\xi)=(\alpha-\tfrac12)t-2\eta t^{1/2}-\sigma \xi t^{1/2}.
\end{equation}
Then, for $\xi>0$,
\begin{equation}
\lim_{t\to\infty} \Pb(x_n(t)\leq x(\xi)) = F_{{\rm GUE}(M)}(\xi+\xi_c),
\end{equation}
and for $\xi<0$,
\begin{equation}
\lim_{t\to\infty} \Pb(x_n(t)\geq x(\xi)) = 0.
\end{equation}
In particular, this means that the distribution function has a Dirac mass at $\xi=0$ with mass
$F_{\rm GUE(M)}(\xi_c)$.
\end{thm}

The situation of Theorem~\ref{ThmMain} is illustrated in Figure~\ref{FigureShock}.
\begin{figure}[t]
\begin{center}
\psfrag{x}[c]{$x$}
\psfrag{n}[l]{$n$}
\psfrag{t12}[cb]{$t^{1/2}$}
\includegraphics[height=5cm]{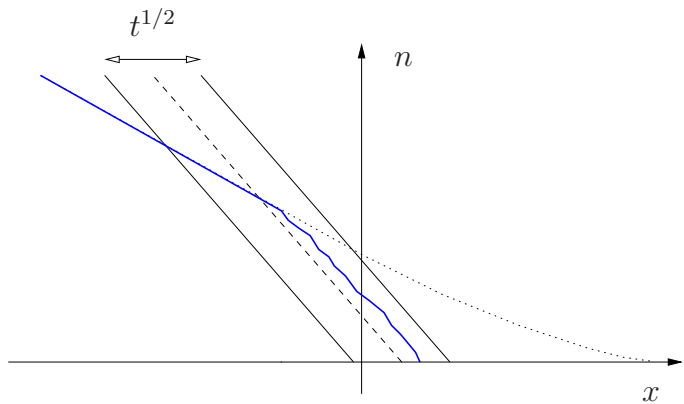}
\caption{Illustration of the shock. The continuous (blue) line is the position of particles, the dashed line is the macroscopic position of particles inside the jam, while the dotted line would be the position without the $M$ $\alpha$-particles. The shock is located in a $t^{1/2}$-neighborhood of the intersection of the dashed and dotted lines.}
\label{FigureShock}
\end{center}
\end{figure}
In the $t^{1/2}$ scale, the fluctuations before the shock are irrelevant since they are only of order $t^{1/3}$. The value $\xi=0$ corresponds to the position which would have particle $n$ if the slow particles were not present. In particular, for $\xi>0$, if $x_n(t)<x(\xi)$ is satisfied, then particle $n$ has already reached the shock. The Dirac mass at $\xi=0$, $F_{\rm GUE(M)}(\xi_c)$, gives the probability that particle $n$ has not yet reached the shock at time $t$. Increasing the value of $\eta$ means looking at a particle which is more to the left and thus it has a larger probability of not having reached the shock at time $t$.

\section{Main ingredients without LPP}

\subsection{Splitting of the problem into two easier ones}\label{SectSplitting}
Instead of studying directly the initial condition (\ref{eqICx}), we consider the following ones:\\
(a) the system without the slow particles,
\begin{equation}\label{eqICA}
x^A_n(0)=-2n, \quad v_n=1,\quad n\geq 1,
\end{equation}
(b) the system with slow particles but with the normal particles initially densely packed,
\begin{equation}\label{eqICB}
x^B_n(0)=-n, \quad n\geq -M+1
\end{equation}
and jump rates (\ref{eqICv}).

Using the graphical construction for TASEP~\cite{Lig76,Har78}, see also~\cite{Li99}, one can couple the processes $\{x_n(t),n\geq -M+1\}$, $\{x^A_n(t),n\geq 0\}$ and $\{x^B_n(t),n\geq -M+1\}$. With this basic coupling, since TASEP is attractive, one immediately has, for any $n\geq 1$,
\begin{equation}
x_n(t)\leq x_n^A(t),\quad x_n(t)\leq x_n^B(t).
\end{equation}
As we shall see, with  a little more of thinking one obtains the following equality.
\begin{lem}\label{lemMinimum}
For all $t\geq 0$ it holds
\begin{equation}\label{eqMinimum}
x_n(t)=\min\{x_n^A(t),x_n^B(t)\},
\end{equation}
for all $n\geq 1$.
\end{lem}
\begin{proof}
In the process with initial conditions (\ref{eqICx}) we introduce a (right-continuous) process $I(t)$ keeping track of the left-most particle which has been affected by the presence of the slow particles. Start with $I(0)=0$. If at time $t$ a jump of particle $m+1$ is blocked by the presence of particle $m$ and $I(t^-)=m$, then we set then $I(t)=m+1$.

At time $t$ there are two possibilities:\\
(a) $I(t)<n$. In this case the presence of the slow particle had not reached particle $n$ already and by the graphical construction we have $x_n(t)=x_n^A(t)$.\\
(b) $I(t)\geq n$. In this case let us show that $x_n(t)=x_n^B(t)$. First of all, $x_1(t)=x_1^B(t)$. At the time $t_1$ when $I$ jumps from $1$ to $2$, we have already $x_1(t_1)=x_1^B(t_1)$. Further, $I$ increases because a jump is suppressed, meaning that $x_2(t_1)=x_1(t_1)-1$. But since by coupling we have always $x_2(t_1)\leq x_2^B(t_1)$, we get
\begin{equation}
x_1(t_1)-1=x_2(t_1)\leq x_2^B(t_1)< x_1^B(t_1)=x_1(t_1).
\end{equation}
This implies that $x_2^B(t_1)=x_1(t_1)-1=x_2(t_1)$ as well. By repeating the argument at each time where $I$ has a jump leads to $x_n(t)=x_n^B(t)$ for all $n\leq I(t)$. See Figure~\ref{FigInfluence} for an illustration.
\begin{figure}[t!]
\begin{center}
\psfrag{x}[cc]{$x$}
\psfrag{t}[cc]{$t$}
\includegraphics[height=13cm]{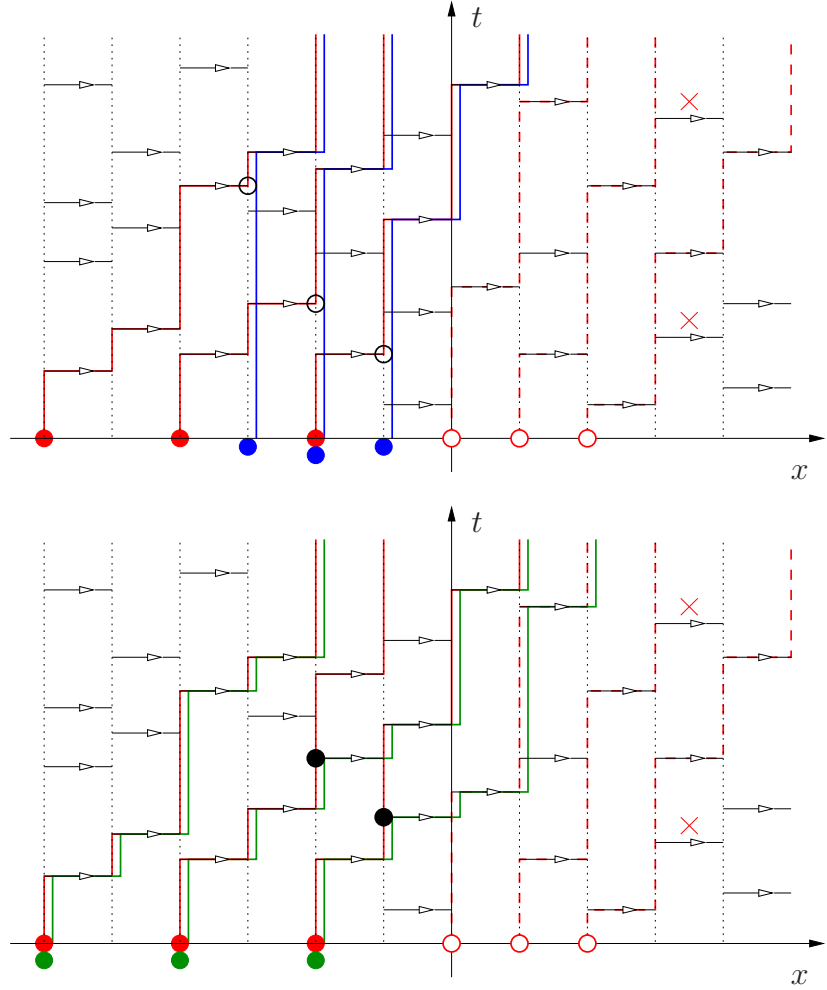}
\caption{The dotted lines are the trajectories of the slow particles. The red crosses represent the fact that with probability $1-\alpha$ the jump of a slow particle did not occur. The red lines are the trajectories of the particles in the original system $\{x_n(t),n\geq 1\}$. The blue lines are the trajectories of $\{x_n^B(t),n\geq 1\}$ and the green lines are the ones of $\{x_n^A(t),n\geq 1\}$. The white dots are the positions after which $x_n^B=x_n$ due to the coupling. The black dots are the points until when $x_n^A=x_n$. At that time, $x_n^A$ can jump but $x_n$ not (they might merge again later). At times where there is a black dot, $t\mapsto I(t)$ increases by one. Between a white and a black dot, the particles involved for the three initial conditions coincide.}
\label{FigInfluence}
\end{center}
\end{figure}
\end{proof}

\subsection{Slow decorrelation}\label{SectSlowDec}
For the proof of Theorem~\ref{ThmMain} we will use the fact that the fluctuations of $x_n^A$ and $x_n^B$ are on two different scales. However, if one would like to prove results for other initial conditions like the ones studied in~\cite{FN13,FN16} without employing the connection to LPP, then one would use slow decorrelation as well. In~\cite{CFP10a} the slow-decorrelation was stated and proven for several models in the KPZ class, for last passage percolation, for positive temperature polymers but also for PASEP. For PASEP show-decorrelation was stated in terms of the height function. Here we explain it in terms of particle positions for the initial condition (\ref{eqICA}).

TASEP with initial condition (\ref{eqICA}) has been studied in~\cite{BFS07}. In particular, compare with (2.21) of~\cite{BFS07}, one considers the rescaled random variable
\begin{equation}
X^{A}_t:=\frac{x^A_{\frac{1-\alpha}{2}t}(t)-(\alpha-\tfrac12)t}{-t^{1/3}}.
\end{equation}
The large $t$ limit of $X^A_t$ is known~\cite{BFS07} (see~\cite{BR99,BR99b,Jo03b} for discrete time analogues in LPP framework).
\begin{lem}\label{lemGOE}
For any $\alpha<1/2$, it holds
\begin{equation}
\lim_{t\to\infty} X^{A}_t\stackrel{(d)}{=}  \tfrac12 \xi_{\rm GOE},
\end{equation}
where $\xi_{\rm GOE}$ is a GOE Tracy-Widom distributed random variable.
\end{lem}
Physically one expects that particles which are around position $(\alpha-\tfrac12)t$ at time $t$ are non-trivially correlated with particles at previous time if they are on (or close to) a given characteristic line. For TASEP with density $\rho$, the characteristic lines have speed $1-2\rho$. In our case $\rho=1/2$, thus we have to look at particles also around position $(\alpha-\tfrac12)t$. Take any $\nu\in (0,1)$. Then a simple computation using (2.21) of~\cite{BFS07} gives
\begin{equation}
\widetilde X^{A}_t:=\frac{x^A_{\frac{1-\alpha}{2}t-t^\nu/4}(t-t^\nu)-(\alpha-\tfrac12)t}{-t^{1/3}}\to  \tfrac12 \widetilde\xi_{\rm GOE}
\end{equation}
in distribution as $t\to\infty$, with $\widetilde \xi_{\rm GOE}$ also a GOE Tracy-Widom distributed random variable~\cite{TW96}. Slow-decorrelation is the following statement.
\begin{prop}
For any $\e>0$,
\begin{equation}
\lim_{t\to\infty}\Pb(|\widetilde X^{A}_t-X^{A}_t|\geq \e)=0.
\end{equation}
\end{prop}
\begin{proof}
The proof is essentially as in~\cite{CFP10b}, except that here we consider particle positions as observables instead of the height function. For the proof, we need one more ingredient. Let $x^{\rm step}_n(t)$ be the position of particle $n$ at time $t$ with step initial condition, i.e., $x^{\rm step}_n(0)=-n+1$, $n\geq 1$. It is well-known~\cite{Jo00b} (see~\cite{BF07} for an approach without using last passage percolation) that
\begin{equation}
X^{\rm step}_{t^\nu}:=\frac{x^{\rm step}_{t^\nu/4}(t^\nu)}{-t^{\nu/3}}\to 2^{-1/3}\xi_{\rm GUE},
\end{equation}
in distribution as $t\to\infty$, where $\xi_{\rm GUE}$ is a GUE Tracy-Widom distribution function~\cite{TW94}.

Now consider the configuration at time $t-t^\nu$, remove all particles to the right of $x^A_{\frac{1-\alpha}{2}t-t^\nu/4}(t-t^\nu)$ and put to its left the remaining $t^\nu/4$ particles densely packed. Then, by basic coupling, we have
\begin{equation}
x^A_{\frac{1-\alpha}{2}t}(t) \leq x^A_{\frac{1-\alpha}{2}t-t^\nu/4}(t-t^\nu)+ x^{\rm step}_{t^\nu/4}(t^\nu)
\end{equation}
with the latter two random variables being independent. After rescaling this becomes
\begin{equation}
X^A_t\leq \widetilde X^A_t + t^{(\nu-1)/3}X^{\rm step}_{t^\nu}.
\end{equation}
In particular, there exists a compensator $Z_t\geq 0$ such that
\begin{equation}
X^A_t= \widetilde X^A_t + t^{(\nu-1)/3}X^{\rm step}_{t^\nu}- Z_t.
\end{equation}
We have $t^{(\nu-1)/3}X^{\rm step}_{t^\nu}\to 0$ in probability, $X^A_t$ and $\widetilde X^A_t$ converging in distribution to the same limit. By Lemma~4.1 of~\cite{BC09} this implies that $Z_t$ converges in probability to zero. Together with the fact that $t^{(\nu-1)/3}X^{\rm step}_{t^\nu}\to 0$ in probability, it implies that
$X^A_t-\widetilde X^A_t$ converges in probability to zero as well.
\end{proof}

\subsection{Localization of the correlations}\label{SectAsymptIndep}
To understand the asymptotic independence, one has to understand which space-time regions are actually relevant for the position of the particles. For that purpose, we need to see how $x_N(t)$ is influenced by the position of previous particles by suppressing jumps.

Let us define the following process running \emph{backwards} in time. Let $N(t)=N$. If at time $\tilde t^+$ we have $N(\tilde t^+)=n$ and at time $\tilde t$ a jump of particle $n$ is suppressed by the presence of particle $n-1$, then we set $N(\tilde t)=n-1$. Further, for any $u \in (0,t)$, we define the particle system $\{\tilde x_n(s),u\leq s\leq t, n\geq N(u)\}$ by setting the position of particles at $s=u$ as $\tilde x_n(u)=x_{N(u)}(u)-n+N(u)$, for $n\geq N(u)$, see Figure~\ref{FigSpaceCut}. The evolution of $\tilde x_n$'s and $x_n$ are coupled by the basic coupling.
\begin{figure}[t!]
\begin{center}
\psfrag{x}[lc]{$x$}
\psfrag{t}[lb]{$t$}
\psfrag{tau}[lb]{$u$}
\includegraphics[height=7cm]{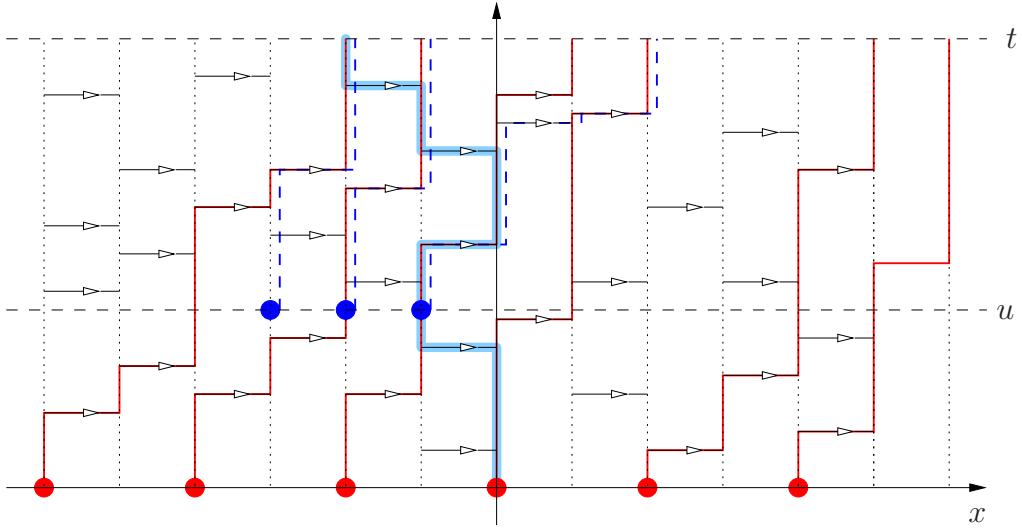}
\caption{The red solid lines are the trajectories of $\{x_n(s),0\leq s\leq t, \textrm{ all }n\}$. The thick light-blue line is the trajectory of $\{x_{N(s)}(s),0\leq s\leq t\}$. The blue dots and the dashed blue lines are the trajectories of $\{\tilde x_{n}(s), u\leq s\leq t,n\geq N(u)\}$.}
\label{FigSpaceCut}
\end{center}
\end{figure}

\begin{prop}\label{prop3.4}
For any $0\leq u\leq t$, we have the identity
\begin{equation}
x_N(t)=\tilde x_N(t) = x_{N(u)}(u)+\left[\tilde x_{N}(t)-\tilde x_{N(u)}(u)\right].
\end{equation}
In particular, one has
\begin{equation}\label{eq3.15}
\tilde x_{N}(t)-\tilde x_{N(u)}(u) \stackrel{(d)}{=} x^{\rm step}_{N-N(u)+1}(t-u)
\end{equation}
where $x^{\rm step}_n(t)$ is the position of particle $n$ at time $t$ with initial condition $x_n^{\rm step}(0)=-n+1$, $n\geq 1$.
\end{prop}
\begin{proof}
Initially we have $\tilde x_{N(u)}(u)=x_{N(u)}(u)$. Let $\tilde t$ the first time where $N(\tilde t)=N(u)+1$. Then, we have $\tilde x_{N(s)}(s)=x_{N(s)}(s)$ for $s\in [u,\tilde t)$. Indeed, by construction $\tilde x_{N(s)}(s)\geq x_{N(s)}(s)$ for $s\in[u,\tilde t)$. Having $\tilde x_{N(s)}(s)> x_{N(s)}(s)$ would imply that for some time $v\in[u,s]$, there is a jump of $x_{N(v)}$ which is suppressed (and not suppressed for $\tilde x_{N(v)}$). But this would imply a change of $N(u)$ by $-1$, which is a contradiction.

Secondly, at time $\tilde t$, $x_{N(\tilde t)}(\tilde t)=x_{N(\tilde t)-1}(\tilde t)-1$ and a jump was suppressed, i.e., $N(\tilde t)=N(\tilde t^+)-1$. This implies
\begin{equation}
x_{N(\tilde t)-1}(\tilde t)-1 = x_{N(\tilde t)}(\tilde t) \leq \tilde x_{N(\tilde t)}(\tilde t) <\tilde x_{N(\tilde t)-1}(\tilde t) = x_{N(\tilde t)-1}(\tilde t),
\end{equation}
and thus we have $\tilde x_{N(\tilde t)}(\tilde t) = x_{N(\tilde t)}(\tilde t)$. Repeating the argument iteratively, we obtain $\tilde x_{N(t)}(t)= x_{N(t)}(t)=x_N(t)$. Finally, the identity (\ref{eq3.15}) follows by the definition of the process $\tilde x$.
\end{proof}

By construction of $N(t)$, the position $x_N(t)$ is equal $\tilde x_N(t)$. Further, notice that if to the left (resp.\ strictly to the right) of the (random) trajectory ${\cal X}=\{x_{N(s)}(s),0\leq s\leq t\}$ the rate of the Poisson processes are replaced by infinite rates, then the position of $x_N(t)$ remains unchanged. The next theorem tells us that ${\cal X}$ is included in a \emph{deterministic} region of size $\Or(t^{2/3+\e})$ with high probability.
\begin{thm}\label{thmcorrelations}
Fix $0<\e<1/3$, $\nu>1$, and set $N(t)=\nu t$. Then, for all $t$ large enough,
\begin{equation}
\Pb\left(\left|x_{N(\tau t)}(\tau t)-(-2\nu+1/2)t\right|\geq 3 t^{2/3+\e}\textrm{ for all }\tau\in[0,1]\right)\leq C e^{-c t^{2\e}}
\end{equation}
for some constants $C,c\in (0,\infty)$.
\end{thm}
This means that the probability that the position $x_N(t)$ depends on the realization of the Poisson processes in
\begin{equation}
\{(j,s), |j-(-2\nu+1/2)t|\geq 3 t^{2/3+\e}, 0\leq s\leq t\}
\end{equation}
is bounded by $C e^{-c t^{2\e}}$.

\subsubsection*{Proof of the localization}
In order to control $N(s)$ and $x_{N(s)}(s)$, we need to have estimates on the position of particles $x_{\nu t}(t)$. One of the tail of its distribution is obtained by comparing with the flat initial condition, $\{x^{\rm flat}_n(0)=-2n, n\in\Z\}$, while the other tail by an estimate on the step initial condition $\{x^{\rm step}_n(0)=-n+1,n\geq 1\}$.
\begin{lem}\label{lemBounds1}
For any $\nu\geq 0$, we have
\begin{equation}\label{eq3.19a}
\Pb(x_{\nu t}(t)\leq -2\nu t+t/2-s t^{1/3})\leq\Pb(x^{\rm flat}_{t/4}(t)\leq -s t^{1/3}),
\end{equation}
and, for any $\nu\geq 1/4$, we have
\begin{equation}\label{eq3.19b}
\Pb(x_{\nu t}(t)\geq -2\nu t+t/2-s t^{1/3})\leq \Pb(x^{\rm step}_{t/4}(t)\geq -s t^{1/3}).
\end{equation}
\end{lem}
\begin{proof}
By the basic coupling, we immediately have that
\begin{equation}\label{eq3.17}
x_{\nu t}(t) \geq x^{\rm flat}_{\nu t}(t) \stackrel{(d)}{=}x^{\rm flat}_{t/4}(t)-2\nu t+t/2.
\end{equation}
This gives (\ref{eq3.19a}). Secondly, for $\nu t\geq t/4$, consider TASEP with initial condition $\{x^{C}_n(0)=-2\nu t+t/2-n+1,n\geq 1\}$, i.e., start with the initial condition $x_n(0)$ from which one removes the particles to the right of $-2\nu t+ t/2$ and move the particles to the left of this position to the right producing a step initial condition ending at $-2\nu t+t/2$. In particular, particle with label $\nu t$ in the original process is moved to $x_{t/4}^C(0)$. Thus we have
\begin{equation}\label{eq3.18}
x_{\nu t}(t)\leq x^C_{t/4}(t) \stackrel{(d)}{=} -2\nu t+t/2+x^{\rm step}_{t/4}(t).
\end{equation}
Equations (\ref{eq3.17}) and (\ref{eq3.18}) implies (\ref{eq3.19b}).
\end{proof}
Bounds on the probabilities of the r.h.s.\ of (\ref{eq3.19a})-(\ref{eq3.19b}) are known and we report them here.
\begin{lem}\label{lemBound2}
Let $\nu \in (0,1)$. There exists a $t_0\in (0,\infty)$ such that for all $t\geq t_0$,
\begin{equation}\label{eq3.20}
\begin{aligned}
&\Pb(x^{\rm step}_{\nu t}(t)\geq (1-2\sqrt{\nu})t-s t^{1/3})\leq C_1\, e^{-c_1 (-s)^{3/2}},\quad s\leq 0,\\
&\Pb(x^{\rm step}_{\nu t}(t)\leq (1-2\sqrt{\nu})t -s t^{1/3})\leq C_2 \,e^{-c_2 s},\quad s\geq 0,\\
&\Pb(x^{\rm flat}_{t/4}(t)\leq -s t^{1/3})\leq C_3 \,e^{-c_3 s},\quad s\geq 0,
\end{aligned}
\end{equation}
where the constants $C_i,c_i$ are positive and independent of $s$. Further, for any given $\e>0$, the constants in the bounds for step initial conditions can be chosen independent of $\nu\in [\e,1-\e]$.
\end{lem}
The first estimate in (\ref{eq3.20}) was obtained in~\cite{BFP12} in terms of TASEP height function. The idea is to bound the Fredholm determinant which gives the distribution function of $x^{\rm step}_{t/4}(t)$ by the exponential of the trace of the kernel, see Section~4 of~\cite{BFP12}. The method was used before by Widom in~\cite{Wid02}. The other two estimates in (\ref{eq3.20}) follow directly from the exponential estimates on the correlation kernel for step initial condition, and for flat initial condition, see for instance in Proposition 5.3 of~\cite{BF07}.

\begin{cor}\label{corBounds}
As a consequence of Lemmas~\ref{lemBounds1} and~\ref{lemBound2}, for any $\nu\geq 1/4$, we have
\begin{equation}\label{eqBounds}
\begin{aligned}
&\Pb(x_{\nu t}(t)\leq -2\nu t+t/2-s t^{1/3})\leq C_3 e^{-c_3 s},\quad s\geq 0\\
&\Pb(x_{\nu t}(t)\geq -2\nu t+t/2-s t^{1/3})\leq C_1 e^{-c_1 (-s)^{3/2}},\quad s\leq 0.
\end{aligned}
\end{equation}
\end{cor}

Further, matching the law of large numbers in Proposition~\ref{prop3.4} we obtains $N(\tau t)\sim (\nu-(1-\tau)/4)t$ for large $t$ and $\tau \in [0,1]$. First, we want to show that the fluctuations of $N(\tau t)$ are in a $t^{2/3+\e}$ region with high probability. From this we will deduce that also the position of $x_{N(\tau t)}(\tau t)$ will be in a region of order $t^{2/3+\e}$ around the characteristic line.

\begin{prop}\label{PropGoodSet}
Let $0<\e<1/3$ and define the good set
\begin{equation}
\Omega_G=\{\omega: |N(\tau t)-(\nu-\tfrac14(1-\tau))t|\leq t^{2/3+\e}\textrm{ for all }0\leq \tau\leq 1\}.
\end{equation}
Then, for all $t$ large enough,
\begin{equation}
\Pb(\Omega_G)\geq 1- \tilde C e^{-\tilde c t^{2\e}}
\end{equation}
for some constants $\tilde C,\tilde c\in (0,\infty)$.
\end{prop}
\begin{proof}
We need to estimate $\Pb(\Omega_G^c)$. If $\omega\in \Omega_G^c$ we can define the last time such that $|N(\tau(\omega)t)-t(\nu-\tfrac14(1-\tau(\omega)))|>t^{2/3+\e}$ and denote it by $\tau(\omega)$. There are two cases which can be analyzed similarly:\\[0.5em]
\emph{Case (a):} $N(\tau(\omega)t)-t(\nu-\tfrac14(1-\tau(\omega)))>t^{2/3+\e}$,\\[0.5em]
\emph{Case (b):} $N(\tau(\omega)t)-t(\nu-\tfrac14(1-\tau(\omega)))<-t^{2/3+\e}$.

\medskip
Consider first \emph{Case (a)}. Since the index process $N$ has one-sided jumps, it means that
$N(\tau(\omega)t)=\lfloor t(\nu-\tfrac14(1-\tau(\omega)))+t^{2/3+\e}\rfloor +1$. For the rest of the estimates, we will never write explicitly the integer parts and also the $+1$ is irrelevant. The goal is to use the estimates of Corollary~\ref{corBounds} and Lemma~\ref{lemBound2}, together with the relation of Proposition~\ref{prop3.4}, to bound the probability of this case.

Let us define $\widetilde N(\tau t)=t(\nu-\tfrac14(1-\tau))+t^{2/3+\e}$ and consider a finite number of times, namely $\tau \in I=\{0,1/t,2/t,\ldots,1\}$. Define the events
\begin{equation}
\begin{aligned}
E_{\tau}&=\{\omega: x_{\widetilde N(\tau t)}(\tau t)\geq -2\nu t+t/2-2 t^{2/3+\e}-\tfrac12 t^{2\e+1/3}\},\\
\widetilde E_{\tau}&=\{\omega: x^{\rm step}_{N-\widetilde N(\tau t)+1}((1-\tau) t)\geq 2 t^{2/3+\e} + \tfrac32 t^{2\e+1/3}\}.
\end{aligned}
\end{equation}
Notice that we need only to consider $\tau\leq 1-4 t^{\e-1/3}$, since otherwise $\widetilde N(\tau t)>\nu t$ and thus $N(\tau t)$ can not take the value $\widetilde N(\tau t)$.

\medskip
\noindent \emph{Bounds on $\Pb(E_\tau)$.}\\
1) For $\tau\leq t^{2\e-2/3}$, i.e., $\tau t\leq t^{2\e+1/3}$, then $\Pb(E_\tau)=1$ since particle $\widetilde N(\tau t$) starts already to the right of $-2\nu t+t/2-2 t^{2/3+\e}-\tfrac12 t^{2\e+1/3}$.\\
2) For $\tau\geq t^{2\e-2/3}$, we can apply Corollary~\ref{corBounds} with $s=\frac{t^{2\e}}{2\tau^{1/3}}$ and obtain $\Pb(E_\tau)\geq 1-C e^{-c t^{2\e}}$ for some constants $C,c$.

\medskip
\noindent \emph{Bounds on $\Pb(\widetilde E_\tau)$.}\\
For $1-\tau\geq 4t^{\e-1/3}$, we can apply Lemma~\ref{lemBound2} with $s=\frac{t^{2\e}(1+3\tau)}{2(1-\tau)^{4/3}}$ and obtain $\Pb(\widetilde E_\tau)\geq 1-C e^{-c t^{2\e}}$ for some constants $C,c$.

\medskip
The above estimates imply
\begin{equation}\label{eq3.28}
\Pb\left(x_{\widetilde N(\tau t)}(\tau t)+x^{\rm step}_{N-\widetilde N(\tau t)+1}((1-\tau) t)\geq -2\nu t+t/2+ t^{2\e+1/3}\right)\geq 1-C e^{-c t^{2\e}}
\end{equation}
for some $C,c\in (0,\infty)$ independent of $\tau$.

Next, for $s\in [0,t]$, there is a $\tau \in I=\{0,1/t,2/t,\ldots,1\}$ with $s-\tau t\in [0,1]$. Define the event
\begin{equation}\label{eq3.29}
G_\tau=\{\omega: |x_{N(s)}(s)-x_{N(\tau t)}(\tau t)|\leq t^{1/3}\textrm{ for some }s\in[\tau t,\tau t+1]\}.
\end{equation}
Then, $\Pb(G_\tau)\geq 1-C e^{-c t^{1/3}}$ since $x_{N(s)}(s)-x_{N(\tau t)}(\tau t)$ is bounded by a Poisson point process with intensity $2$. The same holds for the process with step initial condition. Combining (\ref{eq3.28}) and (\ref{eq3.29}) we obtain
\begin{equation}\label{eq3.30}
\begin{aligned}
&\Pb\left(x_{\widetilde N(\tau t)}(\tau t)+x^{\rm step}_{N-\widetilde N(\tau t)+1}((1-\tau) t)\geq -2\nu t+t/2+ t^{2\e+1/3}\textrm{ for all }\tau\in[0,1]\right)\\
&\geq 1-C t e^{-c t^{2\e}}.
\end{aligned}
\end{equation}
To conclude the proof, we introduce the events
\begin{equation}
\begin{aligned}
\Omega_{B^+}&=\{\omega: N(\tau t)-(\nu-\tfrac14 (1-\tau))t>t^{2/3+\e}\textrm{ for some }0\leq \tau\leq 1\},\\
\Omega_{B^-}&=\{\omega: N(\tau t)-(\nu-\tfrac14 (1-\tau))t<-t^{2/3+\e}\textrm{ for some }0\leq \tau\leq 1\},\\
G^+_\tau&=\{\omega: x_{\widetilde N(\tau t)}(\tau t)+x^{\rm step}_{N-\widetilde N(\tau t)+1}((1-\tau) t)\geq -2\nu t+t/2+ t^{2\e+1/3}\},\\
G^-_\tau&=\{\omega: x_{\widetilde N(\tau t)}(\tau t)+x^{\rm step}_{N-\widetilde N(\tau t)+1}((1-\tau) t)\leq -2\nu t+t/2-t^{2\e+1/3}\},
\end{aligned}
\end{equation}
as well as $\Omega_G^c=\Omega_B=\Omega_{B^+}\cup\Omega_{B^-}$, $F_\pm=\Omega_{B^\pm}^c\bigcap_{\tau\in [0,1]}G_{\tau,\pm}$. We have $\Pb(\Omega_B)\leq \Pb(\Omega_B^+)+\Pb(\Omega_B^-)$ and
\begin{equation}
\Pb(\Omega_B^\pm)\leq \Pb(F_\pm)+\Pb(\cup_{\tau\in[0,1]}G_{\tau,\pm}^c).
\end{equation}
From (\ref{eq3.30}) we have the estimate $\Pb(\cup_{\tau\in[0,1]}G_{\tau,+}^c)\leq 3C t e^{-c t^{2\e}}$. On the other hand, by Proposition~\ref{prop3.4} we have
\begin{equation}
\Pb(F_+) \leq \Pb(x_N(t)\geq -2\nu t+t/4+t^{2\e+1/3})\leq C_1 e^{-c_1 t^{3\e}},
\end{equation}
where the last bound follows from Corollary~\ref{corBounds}. To resume, since $\e\in (0,1/3)$ we have $\Pb(\Omega_B^+)\leq \tilde C e^{-\tilde c t^{2\e}}$ for all $t\geq t_0$.

The bound for \emph{Case (b)} are obtained similarly to the ones of \emph{Case (a)} providing an analogue bound for $\Pb(F_-)$. For instance, instead of $\widetilde N$ we have $\widehat N(\tau t)=t(\nu-\tfrac14(1-\tau))-t^{2/3+\e}$. For $\tau t\leq t^{2\e+1/3}/2$, we can bound the position of the particle $x_{\widehat N(\tau t)}(\tau t)$ by the position without any other particles and gets the required bound.
\end{proof}

Proposition~\ref{PropGoodSet} tells us that $N(\tau t)$ is localized in a $t^{2/3+\e}$ neighborhood of a deterministic value. The second ingredient is the localization of $x_{\nu \tau t}(\tau t)$ from Corollary~\ref{corBounds}.
\begin{prop}\label{propLocalize}
Let $\e>0$ as above. For all $\tau\in [0,1]$ and $\nu\geq 1/4$,
\begin{equation}
\Pb(|x_{\nu t}(\tau t)-(-2\nu+\tau/2) t|\geq t^{2/3+\e})\leq \hat C e^{-\hat c t^{1/3+\e}}.
\end{equation}
\end{prop}
\begin{proof}
For $\tau t>t^{{2/3+\e}}/4$, this follows from the bounds of Corollary~\ref{corBounds}. For $\tau t<t^{2/3+\e}/4$, this follows by the fact that $x_{\nu t}(\tau t)$ is on the right of its initial position, $-2\nu t$ and its number of steps is stochastically dominated by a the ones of a one-sided continuous random walk with rate $1$.
\end{proof}

Now we can prove Theorem~\ref{thmcorrelations}.
\begin{proof}[Proof of Theorem~\ref{thmcorrelations}]
From Propositions~\ref{PropGoodSet} we have that with probability at least $1-\tilde C e^{\tilde c t^{2\e}}$,
\begin{equation}
x_{[\nu-(1-\tau)/4]t-t^{2/3+\e}}(\tau t)\geq x_{N(\tau t)}(\tau t)\geq x_{[\nu-(1-\tau)/4]t+t^{2/3+\e}}(\tau t)
\end{equation}
for all $\tau\in [0,1]$. From Proposition~\ref{propLocalize}, with probability at least $1-\hat C e^{-\hat c t^{1/3+\e}}$,
\begin{equation}
\begin{aligned}
x_{[\nu-(1-\tau)/4]t-t^{2/3+\e}}(\tau t) &\leq -2\nu t+t/2+3 t^{2/3+\e},\\
x_{[\nu-(1-\tau)/4]t+t^{2/3+\e}}(\tau t) &\geq -2\nu t+t/2-3 t^{2/3+\e},
\end{aligned}
\end{equation}
for all $\tau\in [0,1]$, which ends the proof.
\end{proof}

\section{Proof of Theorem~\ref{ThmMain}}\label{SectProofMainThm}
For the proof we employ Lemma~\ref{lemMinimum} on the decomposition of the process in two simpler cases. Their limiting distributions are the content of Propositions~\ref{lemGOEextended} and~\ref{prop4.2}.
\begin{prop}\label{lemGOEextended}
We have
\begin{equation}
\lim_{t\to\infty}\Pb\left(x^A_{\frac{1-\alpha}{2}t+\eta t^{1/2}}(t)\leq (\alpha-\tfrac12)t-2\eta t^{1/2}-s t^{1/3}\right)=F_{\rm GOE}(2s),
\end{equation}
with $F_{\rm GOE}$ denotes the GOE Tracy-Widom distribution function.
\end{prop}
\begin{proof}
It is a slight extension of Lemma~\ref{lemGOE}, which also follows from the limit (2.21) of~\cite{BFS07}.
\end{proof}
The process with initial condition (\ref{eqICB}) has been studied from the LPP/sample covariance matrix point of view already in~\cite{BBP06}. A direct approach is possible also using the Fredholm determinant of the particle of~\cite{BF07}, since for step-initial conditions (and general jump rates) biorthogonalization is explicit and the correlation kernel as well, see Corollary~2.26 of~\cite{BF08}. The result of~\cite{BBP06} is the following.
\begin{prop}\label{prop4.2}
It holds
\begin{equation}\label{eq4.2}
\lim_{t\to\infty}\Pb\left(x^B_{\frac{1-\alpha}{2}t+\eta t^{1/2}}(t)\leq (\alpha-\tfrac12)t-2\eta t^{1/2} -\sigma \xi t^{1/2}\right)=F_{\rm GUE(M)}(\xi+\xi_c),
\end{equation}
with $F_{\rm GUE(M)}$ denotes the distribution of the largest eigenvalue of a $M\times M$ GUE random matrix, and with the constants $\sigma,\xi_c$ as in Theorem~\ref{ThmMain}.
\end{prop}
\begin{proof}
The result has been already established in the sample covariance matrices/LPP picture~\cite{BBP06}. Therefore we indicate here how to get the limiting correlation kernel directly from the formulas for particle positions, but we will not do the details of the asymptotic analysis (which are by now quite standard). The control of the decay for large $s_1,s_2$, which implies the convergence of the Fredholm determinant, and the details on the steep descent paths are not provided here.

From Corollary~2.26 of~\cite{BF08}, for TASEP with particles starting from $y_n(0)=-n$, $n\geq 1$, and with generic jump rates $\alpha_n$, $n\geq 1$, we have
\begin{equation}
\Pb(y_n(t)\geq \xi)=\det(\Id-K_{n,t})_{\ell^2((-\infty,\xi))},
\end{equation}
where the correlation kernel $K_{n,t}$ is given by\footnote{The notation $\frac{1}{2\pi\I}\oint_{\Gamma_I}dz f(z)$ means that the integration $\Gamma_I$ is a simple anticlockwise oriented path including the poles only the poles of $f(z)$ which are in the set $I$.}
\begin{equation}
K_{n,t}(x,y)=\frac{1}{(2\pi\I)^2}\oint_{\Gamma_0}dw \oint_{\Gamma_{\vec \alpha^{-1}}} dz \frac{e^{t/w} w^x}{e^{t/z} z^{y+1}}\frac{\prod_{\ell=1}^n (1-\alpha_\ell w)}{ \prod_{\ell=1}^n (1-\alpha_\ell z)}\frac{1}{w-z}.
\end{equation}
The contour $w$ goes around $0$ and the contour $z$ includes the poles $\alpha_1^{-1},\ldots,\alpha_n^{-1}$.

In our case we have $x_n(t)=y_{n+M}(t)+M-1$ and $\alpha_n=\alpha$ for $n=1,\ldots,M$ and $\alpha_n=1$ for $n\geq M+1$. Thus we have, for $n\geq 1$,
\begin{equation}
\Pb(x_n(t)\geq \xi)=\det(\Id-\widetilde K_{n,t})_{\ell^2((-\infty,\xi))},
\end{equation}
with
\begin{equation}
\begin{aligned}
\widetilde K_{n,t}(x,y)=\frac{1}{(2\pi\I)^2}\oint_{\Gamma_0}dw \oint_{\Gamma_{1,\alpha^{-1}}} dz \frac{e^{t/w} w^{x+1} (1-w)^n}{e^{t/z} z^{y+2} (1-z)^n}\frac{1}{w-z}\left(\frac{w^{-1}-\alpha^{-1}}{z^{-1}-\alpha^{-1}}\right)^M.
\end{aligned}
\end{equation}
We can deform the $w$-contour to include also the pole $w=z$ without changing the result, because the contribution at $w=z$ is the integral over $z$ of $z^{x-y-1}$, which is zero since the $z$-contour does not include $0$.
Due to the scaling (\ref{eq4.2}), we need to determine the $t\to\infty$ limit of the rescaled kernel
\begin{equation}
K^{\rm resc}_t(s_1,s_2):=\sigma t^{1/2} \widetilde K_{n,t}(x(s_1),x(s_2))
\end{equation}
with $x(s)=(\alpha-\tfrac12)t-2\eta t^{1/2}-s \sigma t^{1/2}$.

If we define the functions
\begin{equation}
\begin{aligned}
f_0(w)&=w^{-1}+(\alpha-\tfrac12)\ln(w)+\tfrac{1-\alpha}{2}\ln(1-w),\\
f_1(w,s)&=-(2\eta+s \sigma)\ln(w)+\eta\ln(1-w),
\end{aligned}
\end{equation}
we have
\begin{equation}
K^{\rm resc}_t(s_1,s_2) = \frac{\sigma t^{1/2}}{(2\pi\I)^2}\oint_{\Gamma_{1,\alpha^{-1}}} dz \oint_{\Gamma_{0,z}}dw \frac{e^{t f_0(w)+t^{1/2} f_1(w,s_1)}}{e^{t f_0(z)+ t^{1/2} f_1(z,s_2)}}\frac{w z^{-2}}{w-z}\left(\frac{w^{-1}-\alpha}{z^{-1}-\alpha}\right)^M.
\end{equation}
A simple computation gives that $\frac{df_0(w)}{w}=0$ for $w=\{2,\alpha^{-1}\}$, with $\Re(f_0(\alpha^{-1}))<\Re(f_0(2))$. For the integration contour of $w$ we choose a steep descent path for $\Re f_0(w)$, which passes at a distance $\delta t^{-1/2}$ to the right of $\alpha^{-1}$, and decompose the contribution from $z$ by (a) the one from $1$ and (b) the one from $\alpha^{-1}$.

For (a) one takes a steep descent path for $-\Re f_0(z)$ around $1$ and passing by $2$. This contribution is going to be exponentially small in $t$ since $\Re(f_0(\alpha^{-1})-f_0(2))<0$.

For (b) one takes the contour $|z-\alpha^{-1}|=\delta t^{-1/2}/2$. By doing the change of variables $z=\alpha^{-1}+\alpha^{-1}\sigma^{-1}Z t^{-1/2}$ and $w=\alpha^{-1}+\alpha^{-1}\sigma^{-1}W t^{-1/2}$. Using
\begin{equation}
\begin{aligned}
t f_0(w)&= t f_0(\alpha^{-1})+\frac12 W^2+\Or(W^3 t^{-1/2}),\\
t^{1/2} f_1(w,s)&=t^{1/2} f_1(\alpha^{-1},s) -W(s+\xi_c)+\Or(W^2 t^{-1/2}),
\end{aligned}
\end{equation}
one finally obtains
\begin{equation}
\lim_{t\to\infty} K^{\rm resc}_t(s_1,s_2) \frac{e^{t^{1/2} f_1(\alpha^{-1},s_2)}}{e^{t^{1/2} f_1(\alpha^{-1},s_1)}} = K_{\rm GUE(M)}(s_1+\xi_c,s_2+\xi_c).
\end{equation}
\end{proof}

\begin{proof}[Proof of Theorem~\ref{ThmMain}]
Rescale the random variables $x_n$, $x^A_n$ and $x_n^B$ all in the same way, namely for any fixed $\eta$,
\begin{equation}
Y_t:=\frac{x_{\frac{1-\alpha}{2}t+\eta t^{1/2}}(t)-\left[(\alpha-\tfrac12)t-2\eta t^{1/2}\right]}{-t^{1/2}},
\end{equation}
and similarly $Y_t^A$ and $Y_t^B$.

Since we divide by $t^{1/2}$ and not by $t^{1/3}$, Proposition~\ref{lemGOEextended} implies that
\begin{equation}\label{eq4.4}
\lim_{t\to\infty}\Pb(Y^A_t\leq \xi)\stackrel{(d)}{=}\Id(\xi\leq 0).
\end{equation}
By Lemma~\ref{lemMinimum} we have
\begin{equation}
\Pb\left(Y_t\leq \xi\right) =\Pb\left(\max\{Y^A_t,Y^B_t\}\leq \xi\right).
\end{equation}
For $\xi<0$,
\begin{equation}
\Pb\left(\max\{Y^A_t,Y^B_t\}\leq \xi\right) \leq \Pb\left(Y^A_t\leq \xi\right)\to 0
\end{equation}
as $t\to\infty$ by (\ref{eq4.4}).
For $\xi>0$,
\begin{equation}
\begin{aligned}
&\Pb\left(\max\{Y^A_t,Y^B_t\}\geq \xi\right) \\
=& \Pb\left(\max\{Y^A_t,Y^B_t\}\geq \xi, Y^A_t<Y^B_t\right)+\Pb\left(\max\{Y^A_t,Y^B_t\}\geq \xi, Y^A_t\geq Y^B_t\right)\\
=& \Pb\left(Y^B_t\geq \xi\right)-\Pb\left(Y^B_t\geq \xi, Y^A_t\geq Y^B_t\right)+\Pb\left(Y^A_t\geq \xi, Y^A_t\geq Y^B_t\right).
\end{aligned}
\end{equation}
The last two terms are further bounded by $\Pb\left(Y^A_t\geq \xi\right)$, leading to
\begin{equation}
0\leq \Pb\left(\max\{Y^A_t,Y^B_t\}\geq \xi\right) - \Pb\left(Y^B_t\geq \xi\right)\leq 2 \Pb\left(Y^A_t\geq \xi\right)\to 0
\end{equation}
as $t\to\infty$ by (\ref{eq4.4}).
\end{proof}

\subsection{Comparison with the direct computation}\label{SectComparison}
In this final section we explain why the direct approach of computing the Fredholm determinant in presence of shock is difficult. For simplicity we discuss the $M=1$ case, which corresponds to the situation of the statement in Proposition~1 of~\cite{BFS09}. By Corollary~11 of~\cite{BFS09} we have
\begin{equation}\label{eq4.18}
\Pb(x_n(t)\geq x)=\det(\Id-\chi_x K_{n,t} \chi_x - \chi_x g\otimes f \chi_x)_{\ell^2(\Z)},
\end{equation}
where $\chi_x(y)=\Id_{y<x}$ and
\begin{equation}
\begin{aligned}\label{eqKernel1}
K_{n,t}(x,y)&=\frac{1}{(2\pi\I)^2}\oint_{\Gamma_0}dv \oint_{\Gamma_{0,-v}}\frac{d w}{w}\frac{e^{t w} (w-1)^{n-1}}{w^{x+n-1}}\frac{(1+v)^{y+n-1}}{e^{t(v+1)}v^{n-1}}\frac{(1+2v)}{(w+v)(w-v-1)} \\
f(x)&=\frac{1}{2\pi\I}\oint_{\Gamma_0}dw \frac{e^{t w} (w-1)^{n-1}}{w^{x+n}} \\
g(y)&=\frac{1}{2\pi\I}\oint_{\Gamma_{-1,\alpha-1}}dv \frac{(1+v)^{y+n-1}}{e^{t(v+1)} v^{n-1}}\frac{1+2v}{(v+1-\alpha)(v+\alpha)}.
\end{aligned}
\end{equation}
The idea is to rewrite (\ref{eq4.18}) as
\begin{equation}\label{eq4.20}
\begin{aligned}
\Pb(x_n(t)\geq x)&=\det(\Id-\chi_x K_{n,t} \chi_x) \\
&\times \left(1-\langle g \chi_x, f\rangle-\langle g \chi_x, (\Id-\chi_x K_{n,t} \chi_x)^{-1} \chi_x K_{n,t}\chi_x f\rangle\right)
\end{aligned}
\end{equation}
Due to the scaling (\ref{eq2.6}), let $x(s)=(\alpha-\tfrac12)t-\sigma s t^{1/2}$. Then, for $\xi>0$, one has that the rescaled correlation kernel
\begin{equation}\label{eq4.21}
\left|\sigma t^{1/2} K_{n,t}(x(s_1),x(s_2))\frac{e^{\kappa(x(s_1))}}{e^{\kappa(x(s_2))}} \right|\leq C e^{-c (s_1+s_2)t^{1/6}}
\end{equation}
for all $s_1,s_2\geq 0$, where $e^{\kappa(x(s_1))-\kappa(x(s_2))}$ is a conjugation, and $C,c\in(0,\infty)$ are constants.
This bound implies by standard arguments that $\det(\Id-\chi_x K_{n,t} \chi_x)=1-\Or(e^{-2 c \xi t^{1/6}})$. The scalar product $\langle g \chi_x, f\rangle$ can be computed explicitly. It is a single contour integral and simple steep descent computations leads to the final result (see (4.38) of~\cite{BFS09}). Thus to determine the limit for $\xi>0$ one has to show that the last term in (\ref{eq4.20}) goes to $0$ as $t\to\infty$.

In the sketch of the proof of Proposition~1 it is argued that this is the case since $K_{n,t}\to 0$, see (\ref{eq4.21}). One would think that it is enough to bound the scalar product using Cauchy-Schwarz, i.e., bounding the term by
\begin{equation}
\|g\chi_x\| \frac{1}{1-\|\chi_x K_{n,t} \chi_x\|} \|\chi_x K_{n,t}\chi_x f\|.
\end{equation}
Unfortunately this does not work. The reason is that if we consider the conjugation as in (\ref{eq4.21}) such that $K_{n,t}$ goes to $0$, then one has to use the same conjugation for $g$ and $f$. However, rigorous steep descent analysis can lead to bounds on the conjugated $\|g\chi_x\|$ which diverges as $t\to\infty$. Unfortunately, this divergence is not compensated by a decay in the bound the conjugated $\|\chi_x K_{n,t}\chi_x f\|$.

In principle one could expand $(\Id-\chi_x K_{n,t} \chi_x)^{-1}$ as Neumann series. Then \emph{first} evaluate the products of the functions and operators, leading to a integral representation and \emph{secondly} do the asymptotic analysis on the integrals. For the first term of the series, namely for $\langle g \chi_x, K_{n,t}\chi_x f\rangle$, one has a 4-fold contour integral with several poles. Still it was shown in~\cite{Hin17} that this terms goes to $0$ as $t\to\infty$. The control of all the terms in the Neuman series will become a very tedious analytic task, as the number of integrals will grow (linearly) with the power in the Neumann series. This shows very nicely the usefulness of the decoupling argument used in this and other papers, which is inspired by the physical behavior of the system.


\end{document}